%% file: fmks1012-arxiv.tex
\documentclass[twoside,leqno,twocolumn]{article}  
\usepackage{graphicx,subfigure,float}
\usepackage{algorithmic,algorithm,bbm,amsmath,amssymb}
\usepackage{amsthm}
\usepackage{wrapfig}
\usepackage{color}
\usepackage{xcolor}
\usepackage{fullpage}
\definecolor{Brown}{cmyk}{0, 0.8, 1, 0.6}
\definecolor{Yellow}{rgb}{1, 1, 0}
\definecolor{Light}{gray}{.80}
\definecolor{Dark}{gray}{.20}

\input{my_symbols}

\input{my_thms}

\newcommand{\mynegvspace}{\vspace{-0.05 in}}
\newcommand{\mynegvspacex}{\vspace{-0.25 in}}

\begin{document}

\title{Fast Exact Max-kernel Search}
\author{Ryan Curtin \\
\and 
Parikshit Ram \\
\and 
Alexander G. Gray}
\date{}

\maketitle

\begin{abstract}
\input{abstract}
\end{abstract}
\section{Introduction} \label{sec:intro}
\input{intro}
\section{Max-kernel search} \label{sec:mko}
\input{mko}
\section{Indexing in kernel space \\ with cover trees} \label{sec:ctree}
\input{ctree}
\section{Simple branch-and-bound \\algorithm} \label{sec:bb_algo}
\input{bb_algo}
\section{Runtime analysis} \label{sec:analysis}
\input{analysis}
\section{Evaluation} \label{sec:eval}
\input{eval}
\section{Distributed data} \label{sec:par}
\input{par}
\section{Conclusion} \label{sec:conclusion}
\input{conclusion}
\newpage
\section*{Appendix}
\appendix
\input{app}
{\footnotesize
\bibliography{../../pram_bib_col,../../nns,../../nbd,../../ml,../../manifold,../../de,../../applications}
\bibliographystyle{unsrt}
}
%
\end{document}

%% file: my_symbols.tex
\newenvironment{citemize}{
\begin{list}{\labelitemi}{\leftmargin=1.5em}
  \setlength{\itemsep}{1.5pt}
  \setlength{\parskip}{0pt}
  \setlength{\parsep}{0pt}}{\end{list}
}

\DeclareMathAlphabet{\mathpzc}{OT1}{pzc}{m}{it}

\newcommand{\hilbert}{\mathcal{H}}
\newcommand{\Real}{\mathbb{R}}
\newcommand{\K}{\mathcal{K}}

\newcommand{\dims}{\mathpzc{D}}

\newcommand{\kd}{$kd$}

\newcommand{\ip}[2]{\ensuremath{\langle #1, #2 \rangle}}
\newcommand{\norm}[1]{\ensuremath{\left\| #1 \right\|}}

%% file: my_thms.tex
\newtheorem{theorem}{Theorem}[section]
\newtheorem{lemma}{Lemma}[section]
\newtheorem{definition}{Definition}[section]

%% file: abstract.tex
The wide applicability of kernels makes the problem of max-kernel search ubiquitous and more general than the usual similarity search in metric spaces. We focus on solving this problem efficiently. We begin by characterizing the inherent hardness of the max-kernel search problem with a novel notion of {\em directional concentration}. Following that, we present a method to use an $O(n \log n)$ algorithm to index any set of objects (points in $\Real^\dims$ or abstract objects) \textit{directly in the Hilbert space} without any explicit feature representations of the objects in this space. We present the first provably $O(\log n)$ algorithm for exact max-kernel search using this index. Empirical results for a variety of data sets as well as abstract objects demonstrate up to 4 orders of magnitude speedup in some cases. Extensions for approximate max-kernel search are also presented.

%% file: intro.tex
\noindent
\textbf{Max-kernel search.} In this paper, we present a novel algorithm to accelerate the following general task of max-kernel search (MKS): for a given set $S$ of $n$ objects (the \textit{reference set}), a Mercer kernel function $\K(\cdot, \cdot)$, and a query $q$, find the object $p \in S$ such that:
\begin{equation}\label{eq:mko}
p = \arg \max_{r \in S} \K(q, r).
\end{equation}
This general form of problem is ubiquitous in computer science; it can be easily seen as a similarity search (for some similarity function $\K(\cdot, \cdot)$).  The most simple approach to this general problem is a linear scan over all the objects in $S$. However, the computational cost of this approach scales linearly with the size of the reference set ($|S| = n$) for a single query, making it computationally prohibitive for large data sets.

\begin{figure}[ht]
\mynegvspace
\begin{center}
\includegraphics[width=0.75\columnwidth,clip=true,trim= 1.5in 1.0in 1.5in 0.5in]{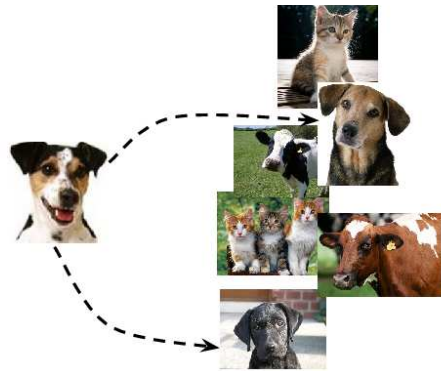}
\end{center}
\mynegvspacex
\mynegvspace
\mynegvspace
\mynegvspace
\caption{Matching images -- an example of MKS with the linear, cosine, Gaussian and polynomial kernels as the usual similarity functions.}
\label{fig:mko-ex1}
\mynegvspace
\mynegvspace
\mynegvspace
\end{figure}

A Mercer kernel can define a measure of similarity for classes of objects including points in $\Real^\dims$, and extending to objects which do not have natural fixed-length representations. Kernels are ubiquitous and can be devised for any new class of objects, such as images and documents (which can be considered as points in $\Real^\dims$), to more abstract objects  like strings (protein sequences \cite{leslie2002spectrum}, text), graphs (molecules \cite{borgwardt2005protein}, brain neuron activation paths), and time series (music, financial data) \cite{muller1997predicting}. The beauty of kernels is the renowned ``kernel trick'' -- the ability to evaluate similarity (inner products) between any pair of objects in some feature space without requiring the explicit feature representations of those objects. 

A special case is the problem of nearest-neighbor search (NNS) in metric spaces, in which the closest object to the query with respect to a distance metric is sought.  However, the requirement of a distance metric makes many efficient methods for exact and approximate NNS~\cite{clarkson2006nearest} unapplicable to the general MKS.

\vspace{0.05in}
\noindent
\textbf{Some large-scale applications.} The widely studied problem of image matching in computer vision is a MKS. The sets of images are of the order of millions and still growing, making linear scan computationally prohibitive. MKS also appears in maximum-a-posteriori inference~\cite{klaas2005fast}, as well as collaborative filtering (via the widely successful matrix factorization framework)~\cite{koren2009matrix}. This matrix factorization obtains an accurate representation of the data in terms of user vectors and item vectors, and the desired result -- the user preference of an item -- is the inner-product between the two respective vectors (a Mercer kernel).  With ever-scaling item sets and millions of users \cite{dror2011yahoo}, reducing the cost of retrieval of recommendations (which is a MKS) is critical for real-world systems. Finding similar protein/DNA sequences for a query sequence from a large set of biological sequences is also an instance of MKS with biological sequences as the objects and various domain-specific kernels (for example, the $p$-spectrum kernel~\cite{leslie2002spectrum}, maximal segment match score \cite{altschul1990basic} \& Smith-Waterman alignment score \cite{smith1981identification}\footnote{The last two functions provide matching scores for pairs of sequences and can easily be shown to be Mercer kernels.}). An efficient MKS algorithm can be directly applied to biological (sub)sequence matching.

\vspace{0.05in}
\noindent
\textbf{Contributions.} To address the breadth of applications of the MKS problem, we present a method to accelerate max-kernel search {\em for any class of objects with a corresponding Mercer kernel}.  Our contributions are:
\begin{citemize}
\item The first concept for characterizing the hardness of MKS in terms of the concentration of the kernel values in any direction -- the {\em directional concentration}.
\item An $O(n \log n)$ algorithm to index any set of objects {\em directly in the Hilbert space defined by the kernel} without requiring explicit representations of the objects in this space.
\item A novel branch-and-bound algorithm on the index in the Hilbert space, which can achieve orders of magnitude speedups over linear search.
\item Value-approximate and order-approximate extensions to the exact max-kernel search algorithm.
\item The first $O(\log n)$ runtime bound for {\em exact} MKS with our proposed algorithm for {\em any} Mercer kernel.
\end{citemize}

\subsection{Related work}
Although there are existing techniques for MKS, \textit{almost all of them solve the approximate problem under restricted settings}. The most common assumption is that the objects are points in some metric space and the kernel function is {\em shift-invariant} -- a monotonic function of the distance between the two objects $(\K(p, q) = f(\norm{p - q}))$, such as the Gaussian radial basis function (RBF) kernel. For shift-invariant kernels, a tree-based recursive algorithm has been shown to scale to large sets for maximum-a-posteriori inference \cite{klaas2005fast}. However, a shift-invariant kernel is only applicable to objects already represented in some metric space.  In fact, the MKS with a shift-invariant kernel is equivalent to NNS in the input space itself, and can be solved directly using existing methods for NNS, an easier and better-studied problem. For shift-invariant kernels, the points can be explicitly embedded in some low-dimensional metric space such that the inner product between the representations of any two points approximates their corresponding kernel value \cite{rahimi2007random}. This step takes $O(n \dims^2)$ time for $S \subset \Real^\dims$ and can be followed by NNS on these representations to accomplish MKS.

Locality-sensitive hashing (LSH) \cite{gionis1999similarity} is widely used for image matching, but only with explicitly representable kernel functions which admit a locality sensitive hashing function \cite{charikar2002similarity}\footnote{The Gaussian and the cosine kernel admit locality sensitive hashing functions with some modifications.}. Kulis and Grauman \cite{kulis2009kernelized} apply LSH to solve MKS {\em approximately} for normalized kernels without any explicit representation. Normalized kernels restrict the self-similarity value to a constant $(\K(x,x) = \K(y,y)\ \forall\ x,y \in S)$. The preprocessing time for this LSH is $O(p^3)$ and a single query requires $O(p)$ kernel evaluations. Here $p$ controls the accuracy of the algorithm -- larger $p$ implies better approximation; the suggest value for $p$  is $O(\sqrt{n})$ with no rigorous approximation guarantees.

A recent work \cite{ram2012maximum} proposed the first technique for exact MKS using a tree-based branch-and-bound algorithm but is restricted only to linear kernels and the algorithm does not have any runtime guarantees. The authors suggest a method for extending their algorithm to non-representable spaces with general Mercer kernels but require $O(n^2)$ preprocessing time.

\vspace{0.05in}
\noindent
\textbf{A case for un-normalized kernels.} While some of the kernels used in machine learning (for example, the Gaussian and the cosine kernel) are normalized, some common kernels like the linear kernel (and the polynomial kernels) are not normalized. Many applications require un-normalized kernels: (1) In the retrieval of recommendations, the normalized linear kernel will result in inaccurate user-item preference scores. (2) In biological sequence matching with the domain-specific matching functions, $\K(x,x)$ implicitly corresponds to presence of (genetically) valuable letters (like W, H, P) or invaluable letters (like X)\footnote{See the score matrix for letter pairs in protein sequences at http://www.ncbi.nlm.nih.gov/Class/FieldGuide/BLOSUM62.txt.} in the sequence $x$. This crucial information is lost in kernel normalization. In this paper, we consider general un-normalized kernels as well as the special case of normalized kernels.

None of the existing techniques can be directly applied to every instance of MKS with general Mercer kernels and any class of objects (Equation~\ref{eq:mko}). Moreover, almost all existing techniques resort to approximate solutions. In this paper, we present a method to perform the exact max-kernel search with $O(n \log n)$ preprocessing time and $O(\log n)$ query time, as well as approximate methods with rigorous accuracy guarantees.

\vspace{0.05in}
\noindent
\textbf{This paper.} In the following section, we contrast MKS to NNS, investigate the inherent hardness of MKS and motivate the applicability of indexing schemes used for NNS in MKS. Section~\ref{sec:ctree} discusses the construction of a tree in the kernel space without explicit representations of objects. Section~\ref{sec:bb_algo} presents the novel branch-and-bound algorithm for exact MKS and extends it to approximate MKS. Sections~\ref{sec:analysis}~and~\ref{sec:eval} examine the theoretical and empirical performance of the approach, respectively. Section~\ref{sec:par} extends the proposed method to the distributed setting and we conclude with Section~\ref{sec:conclusion}.

%% file: mko.tex
A Mercer kernel implies that the kernel value for a pair of objects $(x,y)$ corresponds to an inner-product between the vector representation of the objects $\varphi(x), \varphi(y)$ in some inner-product space $\hilbert$ $(\K(x, y) = \ip{\varphi(x)}{\varphi(y)}_\hilbert)$. Hence every Mercer kernel induces the following metric in $\hilbert$:
\begin{equation} \label{eq:induced-metric}
\begin{split}
d_\K(x, y)  &=  \norm{\varphi(x) - \varphi(y)} \\
 &= \sqrt{\K(x,x) + \K(y,y) - 2\K(x,y)}.
\end{split}
\end{equation}
Whenever MKS can be reduced to NNS in $\hilbert$, NNS methods for general metrics \cite{clarkson2006nearest} can be used for efficient MKS. In this section, we show that this reduction is possible only under strict conditions.

Subsequently, we discuss the hardness of MKS and contrast it to the hardness of NNS. Finally, we discuss desirable properties of certain NNS techniques and their applicability to MKS.

\vspace{0.05in}
\noindent
\textbf{Possible reductions and conditions.}
The nearest neighbor for a query $q$ in $\hilbert$ $\left(\arg \min_{r\in S} d_\K(q,r)\right)$ is the max-kernel candidate (Eq.~\eqref{eq:mko}) if $\K(\cdot, \cdot)$ is a normalized kernel. The two problems can have very different answers with un-normalized kernels. Every kernel also induces a cosine similarity $\left(\K(x,y)/\sqrt{\K(x,x)\K(y,y)}\right)$ in $\hilbert$. Similar to the previous case, the object with the maximum cosine-similarity in $\hilbert$ is the max-kernel candidate only for normalized kernels.

\noindent
\textbf{Hardness of nearest-neighbor search.} However, even if MKS can be reduced to NNS, NNS is still hard for general metrics ($\Omega(n)$ for a single query) without any assumption on the structure of the metric and the set $(S, d)$. Here we present one notion of characterizing the hardness in terms of the structure of the metric \cite{karger2002finding}:
\begin{definition} \label{def:int_dim}
Let $B_S(p,\Delta) = \{r\in S \colon d(p,r) \leq \Delta \}$ denote a closed ball of radius $\Delta$ around some $p\in S$ with respect to a metric $d$. Then, the \textbf{expansion constant} of $(S, d)$ is defined as the smallest $c\geq 2$ such $\left|B_S(p,2\Delta)\right| \leq c\left|B_S(p,\Delta)\right|~\forall p\in S$ and $\forall \Delta>0$.
\end{definition}
The expansion constant effectively bounds the number of points that could be sitting on the surface of a hyper-sphere of any radius around any point. If $c$ is high, NNS could be expensive. A value of $c \sim \Omega(n)$ implies that the search cannot be better than linear scan asymptotically. Under the assumption of bounded expansion constant, NNS methods with sub-linear/logarithmic theoretical runtime guarantees have been presented \cite{karger2002finding, beygelzimer2006cover}.
\subsection{Characterizing the hardness of MKS}
The kernel values for a query are proportional to the length of the projections in the direction of the query in $\hilbert$. While the hardness of NNS depends on how concentrated the surface of spheres are (as characterized by the expansion constant), the hardness of MKS should depend on the distribution of the projections in the direction of the query. This distribution can be characterized using the distribution of points in terms of distances. \textit{If two points are close in distance, then their projections in any direction are close as well. However, if two points have close projections in a direction, it is not necessary that the points themselves are close to each other.} It is to characterize this reverse relationship between points (closeness in projections to closeness in distances) that we present a new notion of concentration in any direction:
\begin{definition} \label{def:direction-conc-constant}
Let $\K(x,y) = \ip{\varphi(x)}{\varphi(y)}_\hilbert$ be a Mercer kernel, $d_\K(x, y)$ be the induced metric from $\K$ and $B_S(p,\Delta)$ denote the corresponding closed ball of radius $\Delta$ around a point $p$ in $\hilbert$. Let $I_S(v, [a, b]) = \{r \in S \colon \ip{v}{\varphi(r)}_\hilbert \in [a,b] \}$ be the set of objects in $S$ projected onto a direction $v$ in $\hilbert$ lying in the interval $[a,b]$ along $v$. Then, the \textbf{directional concentration constant} of $(S, \K)$ is defined as the smallest $\gamma \geq 1$ such that $\forall u \in \hilbert$ such that $\norm{u}_\hilbert = 1$, $\forall p\in S$ and $\forall \Delta > 0$, the set $I_S(u, [\ip{u}{\varphi(p)}_\hilbert - \Delta, \ip{u}{\varphi(p)}_\hilbert + \Delta])$ can be covered by at most $\gamma$ balls of radius $\Delta$.
\end{definition}
\begin{figure*}[t]
\mynegvspace
\begin{center}
\subfigure[Projection interval set.]{ \label{fig:proj-int}
\includegraphics[width=0.31\textwidth,clip=true,trim= 2.0in 6.3in 2.0in 1.7in]{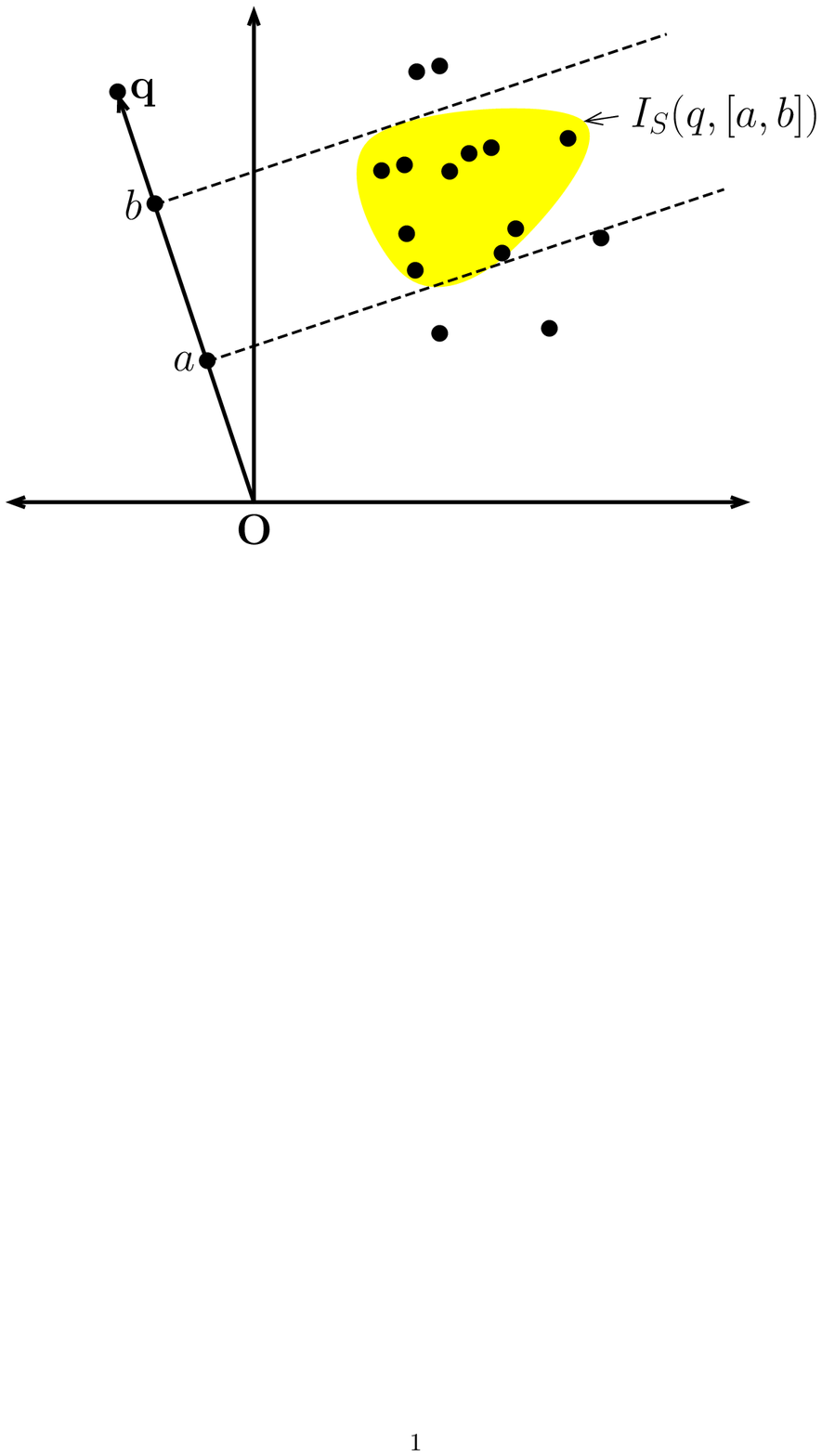}
}
\subfigure[Low value of $\gamma$.]{ \label{fig:low-dcc}
\includegraphics[width=0.31\textwidth,clip=true,trim= 1.8in 6.8in 1.9in 1.6in]{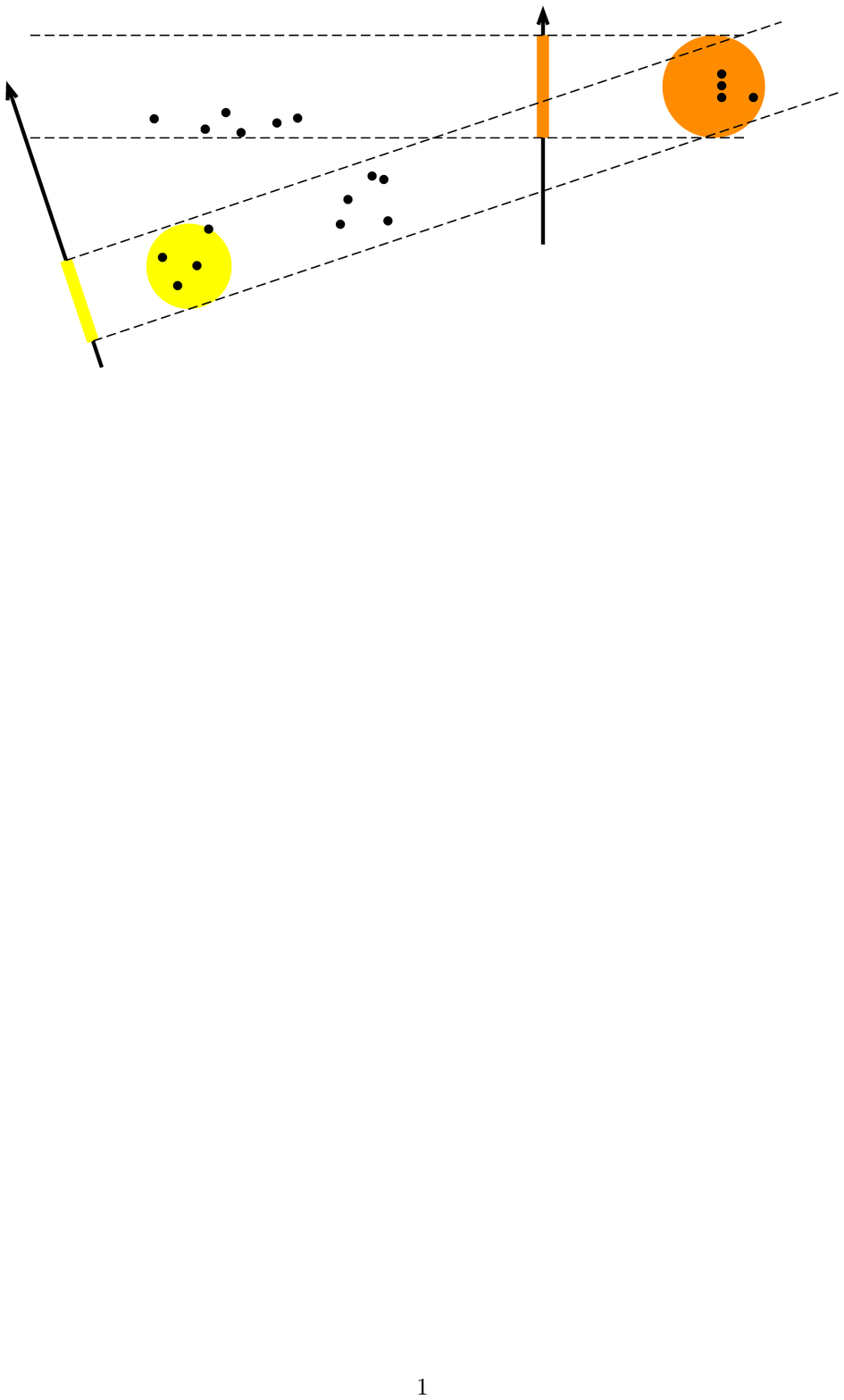}
}
\subfigure[High value of $\gamma$.]{ \label{fig:high-dcc}
\includegraphics[width=0.31\textwidth,clip=true,trim= 1.5in 6.5in 1.6in 2.0in]{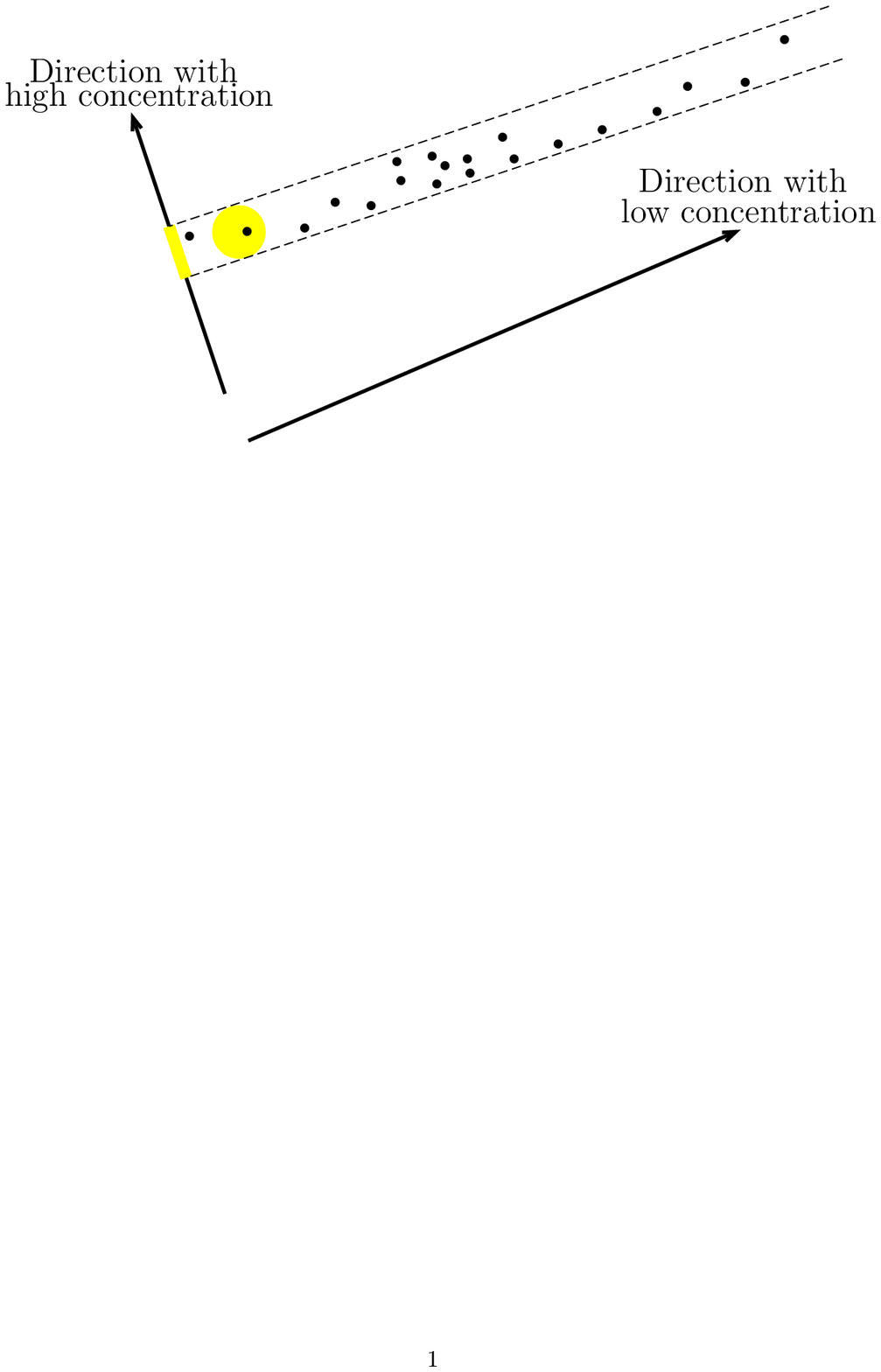}
}
\end{center}
\label{fig:proj-conc}
\mynegvspacex
\caption{Concentration of projections.}
\mynegvspace
\mynegvspace
\end{figure*}
The directional concentration constant is not a measure of the number of points projected into a small interval, but rather a measure of the number of ``patches'' of the data in a particular direction. For a set of points to be close in projections, there are at most $\gamma$ subsets of points which are close in distances as well. Consider the set of points $B = I_S(q, [a, b])$ projected into an interval in some direction (Figure \ref{fig:proj-int}). A high value of $\gamma$ implies that the number of points in $B$ is high but the points are spread out and the number of balls (with diameter equal to $|b-a|$) required to cover all these points is high as well (each point possibly requiring an individual ball). Figure \ref{fig:high-dcc} provides one such example. A low value of $\gamma$ implies that either $B$ has a small size or the size of $B$ is high and $B$ can be covered with a small number of balls (of diameter $|b-a|$).  Figure \ref{fig:low-dcc} is an example of a set with low $\gamma$. The directional concentration constant bounds the number of balls of a particular radius required to index points that project onto an interval of size twice the radius.
\subsection{Desirable existing techniques}
The notion of bounded directional concentration implies that indexing schemes capable of efficiently indexing points in a particular direction might be useful for the task of MKS.
Indexing schemes like space-partitioning trees have been widely used for NNS with success in many cases.  Trees offer good hierarchical indexing schemes in low to medium dimensions; for high-dimensional data, trees which exploit some low-dimensional structure have been developed \cite{dasgupta2008random, beygelzimer2006cover}.  These hierarchical indexing schemes lend easily to intuitive branch-and-bound algorithms for efficient solutions (especially approximate solutions \cite{ciaccia2000pac, ram2009rank}). In addition, tree-based branch-and-bound algorithms are essentially incremental algorithms, and thus easily extend to anytime algorithms \cite{ram2012nearest}. Most importantly, trees only require a single construction -- different algorithms can work on the same tree; in addition, once a tree is built, point insertions and deletions are generally cheap.
%
%

\vspace{0.05in}
\noindent
\textbf{Which tree to use for MKS?} Given the numerous advantages of trees, we wish to select an appropriate tree for efficient MKS. The following two properties are desired of any indexing scheme used for MKS:
\mynegvspace
\begin{citemize}
\item Explicit representation of the objects is not required.
\item It should have sub-quadratic construction time.
\end{citemize}
\mynegvspace
The \kd-tree \cite{friedman1977algorithm} and the metric tree
\cite{preparata1985computational} exhibit good characteristics and are widely used in NNS. However, the \kd-tree requires the explicit representation of the points $\varphi(x)$ in $\hilbert$ for its axis-aligned splits. For similar reasons, random-projection trees~\cite{dasgupta2008random} and PCA-trees \cite{mcnames2001fast} cannot be used for MKS. Metric trees can be constructed without the explicit representations since they can work with only the ability to evaluate the induced metric $d_\K$~\cite{moore2000anchors}\footnote{Calculation of the explicit mean $\mu$ of a node is avoidable by using the kernel trick for operations on the mean: $\left(\ip{\mu}{\varphi(q)} = \sum_{x \in T} \K(x, q) / |T| \right)$ (where $T$ is the set of points in the node). However, this makes the tree construction and tree search computational prohibitive for efficient MKS.}. In this paper, we consider the recently formulated {\em cover tree}~\cite{beygelzimer2006cover} for MKS. It provides a systematic way to build a tree without explicit object representations and with a sub-quadratic construction time. A key difference between cover trees and the aforementioned trees is that the \kd-tree and the metric tree are binary trees while the cover trees can have multiple number of children. Moreover, the time complexities of building and querying a cover trees have been analyzed extensively \cite{beygelzimer2006cover, ram2009linear}, whereas \kd-trees and metric trees have been analyzed only under very limited settings \cite{friedman1977algorithm}. 

%% file: ctree.tex
A cover tree stores a dataset $S$ of size $n$ in the form of a levelled tree. The structure has $O(n)$ space requirement (Theorem 1 in \cite{beygelzimer2006cover}). Each level is a ``cover'' for the level beneath it and is indexed by an integer scale $i$ which decreases as the tree is descended. Let $C_i$ denote the set of nodes at scale $i$. For all scales $i$, the following invariants hold: (i) (nesting invariant) $C_i \subset C_{i - 1}$, (ii) (covering tree invariant)  $\forall p \in C_{i - 1}, \exists q \in C_i \colon d(p, q) \leq 2^i$, and exactly one such $q$ is a parent of $p$. (iii) (separation invariant) For all $p, q \in C_i$, $d(p, q) > 2^i$.

Although the cover tree is defined as infinite levelled sets, the tree has a precise finite representation. As explained in \cite{ram2009linear}, ``The \emph{implicit representation} consists of infinitely many levels $C_i$ with the level $C_\infty$ containing a single node which is the root and the level $C_{-\infty}$ containing every point in the dataset as a node. The \emph{explicit representation} is required to store the tree in $O(n)$ space. It coalesces all nodes in the tree for which the only child is the self-child. This implies that every explicit node either has a parent other than the self-parent or has a child other than a self-child.''

\vspace{0.05in}
\noindent
\textbf{Tree construction in the kernel space $\hilbert$.}  Every node in a cover tree is associated with a single point $p$. Hence, the cover tree construction only requires distances between the points in the set (the tree construction algorithm is presented in Section A of the supplement); and therefore construction in $\hilbert$ \textit{does not require any explicit representation in} $\hilbert$. From Equation~\ref{eq:induced-metric}, the pairwise distance $d_\K(x, y)$ can be evaluated with only three kernel evaluations. If the self-kernel values $\K(x, x)\ \forall\ x \in S$ are precomputed and saved, $d_\K(x, y)$ only requires one evaluation of $\K(x, y)$. The tree construction time is bounded by the following theorem:
\vspace{0.01in}
\begin{theorem} \label{thm:construction-time}
(Theorem 3.6 in \cite{beygelzimer2006coverLonger}) For a data set $S$ of $n$ objects and a metric $d_\K$ with expansion constant $c$, the tree construction requires at most $O(c^6 n \log n)$ time.
\end{theorem}
\noindent
\textbf{Remark.} We would like to point out that while we consider the cover
tree to implicitly index points in $\hilbert$ without any modification, the
cover tree has only been used for NNS to this date. The novelty of this paper is a new branch-and-bound {\em algorithm} using this tree (Section \ref{sec:bb_algo}) to solve the general task of MKS with provable theoretical guarantees (Section \ref{sec:analysis}) and supporting empirical evidence (Section \ref{sec:eval}).

%% file: bb_algo.tex
In this section, we present a simple branch-and-bound algorithm on the cover tree for MKS. Branch-and-bound is widely used in NNS with the help of the triangle inequality of the distance metric. In the absence of the triangle inequality for kernel functions, we obtain a novel bound on the max-kernel value possible between a query and any subtree of a cover tree. Then we present the algorithm for exact and approximate search.

\vspace{0.05in}
\noindent
\textbf{Bounding the max-kernel value.} A cover tree node is defined by an object $p$ and a level $i$. 
Let $S_p^i$ denote the set of objects in the subtree rooted at a node defined by object $p$ at level $i$. In the following theorem, we bound the maximum possible kernel value between a query and an object in the subtree of the cover tree. For notational convenience, we will denote $\max_{r \in R} \K(q, r)$ as $\K(q, R)$.
\begin{theorem} \label{thm:single-ball-bnd}
\mynegvspace
Given a cover tree node rooted at an object $p$ at level $i$ in the kernel space $\hilbert$ and a (query) object $q$, the maximum kernel function value between $q$ and any object in the set $S_p^i$ is bounded as:
\mynegvspace
\begin{equation} \label{eq:mko-ball-bnd}
\K(q, S_p^i) \leq \K(q, p) + 2^{i+1} \sqrt{\K(q, q)}.
\mynegvspace
\end{equation}
\mynegvspace
\mynegvspace
\mynegvspace
\mynegvspace
\end{theorem}
\begin{proof}
Suppose that $p^*$ is the best possible match in the set $S_p^i$ for $q$ and let
$\vec{u}$ be a unit vector in the direction of the line joining $\varphi(p)$ to
$\varphi(p^*)$ in $\hilbert$. Then $\varphi(p^*) = \varphi(p) + \Delta \cdot
\vec{u}$ where $\Delta = d_\K(p, p^*)$ is the distance $\varphi(p)$ and the best
possible match $\varphi(p^*)$ (see Figure \ref{fig:pt-ball-bnd}). Then we have
the following:
\begin{eqnarray}
\mynegvspace
\K(q, S_p^i) & = & \K(q, p^*) = \ip{\varphi(q)}{\varphi(p^*)}_\hilbert  \nonumber \\
    & = & \ip{\varphi(q)}{\varphi(p) + \Delta \cdot \vec{u}}_\hilbert \nonumber \\
\label{eq:first-mko-ineq} & \leq & \ip{\varphi(q)}{\varphi(p)}_\hilbert + \Delta \norm{\varphi(q)}_\hilbert,
\mynegvspace
\end{eqnarray}
where the last inequality follows from the Cauchy-Schwartz inequality $(\ip{x}{y} \leq \norm{x}\norm{y})$ and the fact that $\norm{\vec{u}} = 1$. From the definition of the kernel function, Equation \ref{eq:first-mko-ineq} gives us  $\K(q, S_p^i) \leq \K(q, p) + \Delta \sqrt{\K(q, q)}.$
We bound $\Delta$ from above using the covering invariant -- for any cover tree node $p$ at level $i$, the distance to the farthest child node is bounded by $2^i$. Applying this bound recursively with the triangle inequality of $d_\K$ gives us $\Delta = d_\K(p, p^*) \leq \sum_{j = -\infty}^i 2^j = 2^{i + 1}$. The statement of the theorem follows.
\end{proof}
\mynegvspace
\begin{figure}[ht]
\begin{center}
\includegraphics[width=\columnwidth,clip=true,trim= 1.5in 6.2in 1.5in 1.7in]{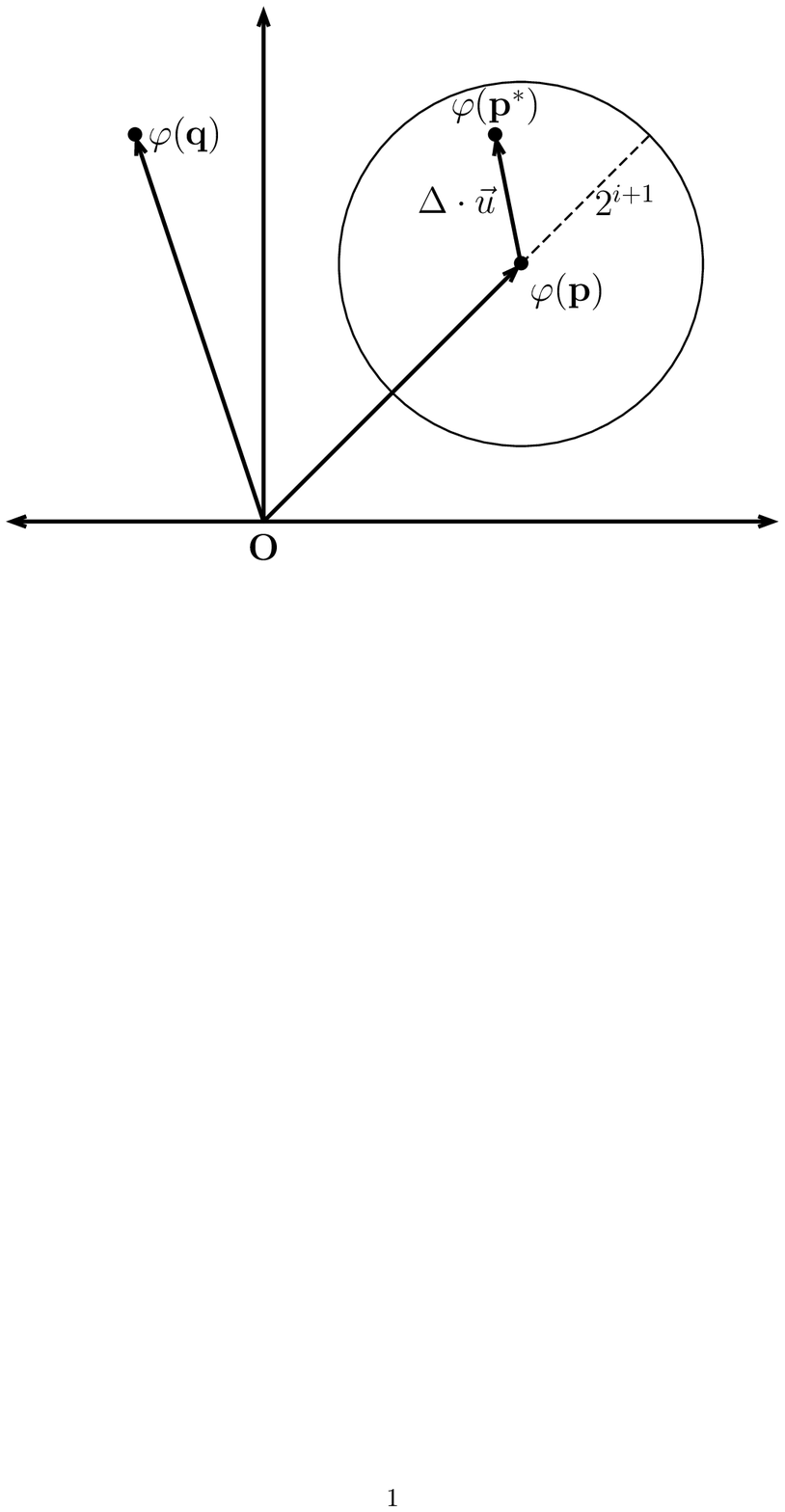}
\end{center}
\mynegvspacex
\mynegvspace
\caption{Max-kernel upper bound.}
\label{fig:pt-ball-bnd}
\mynegvspace
\mynegvspace
\mynegvspace
\end{figure}
\begin{figure*}[t]
\mynegvspacex
\begin{center}
\subfigure[{\small Best possible kernel value}]{ \label{fig:kernel-proj}
\includegraphics[width=0.31\textwidth,clip=true,trim= 1.5in 5.5in 1.7in 1.7in]{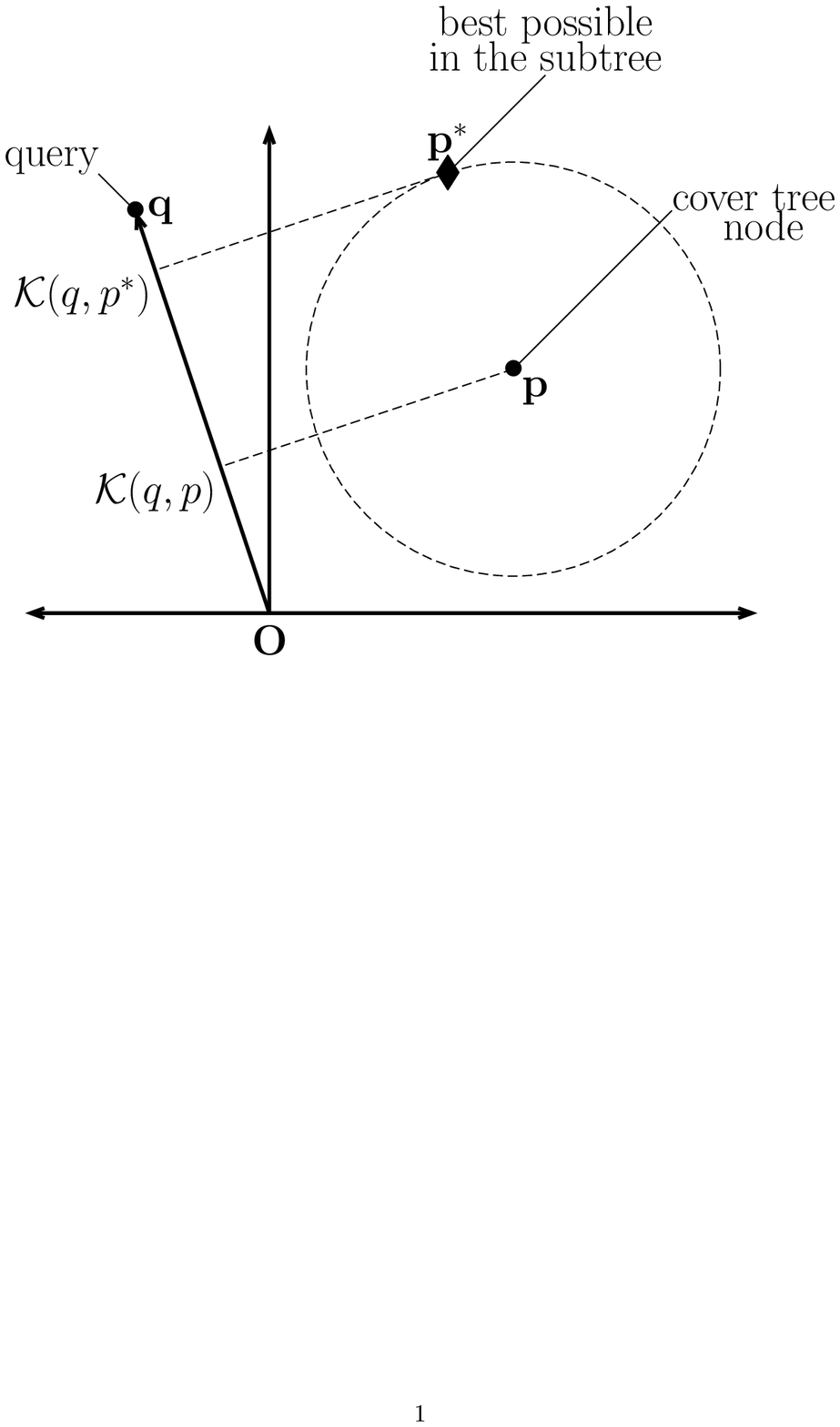}
}
\subfigure[{\small At level $i$}]{ \label{fig:prune1}
\includegraphics[width=0.31\textwidth,clip=true,trim= 1.5in 6.17in 1.7in 1.7in]{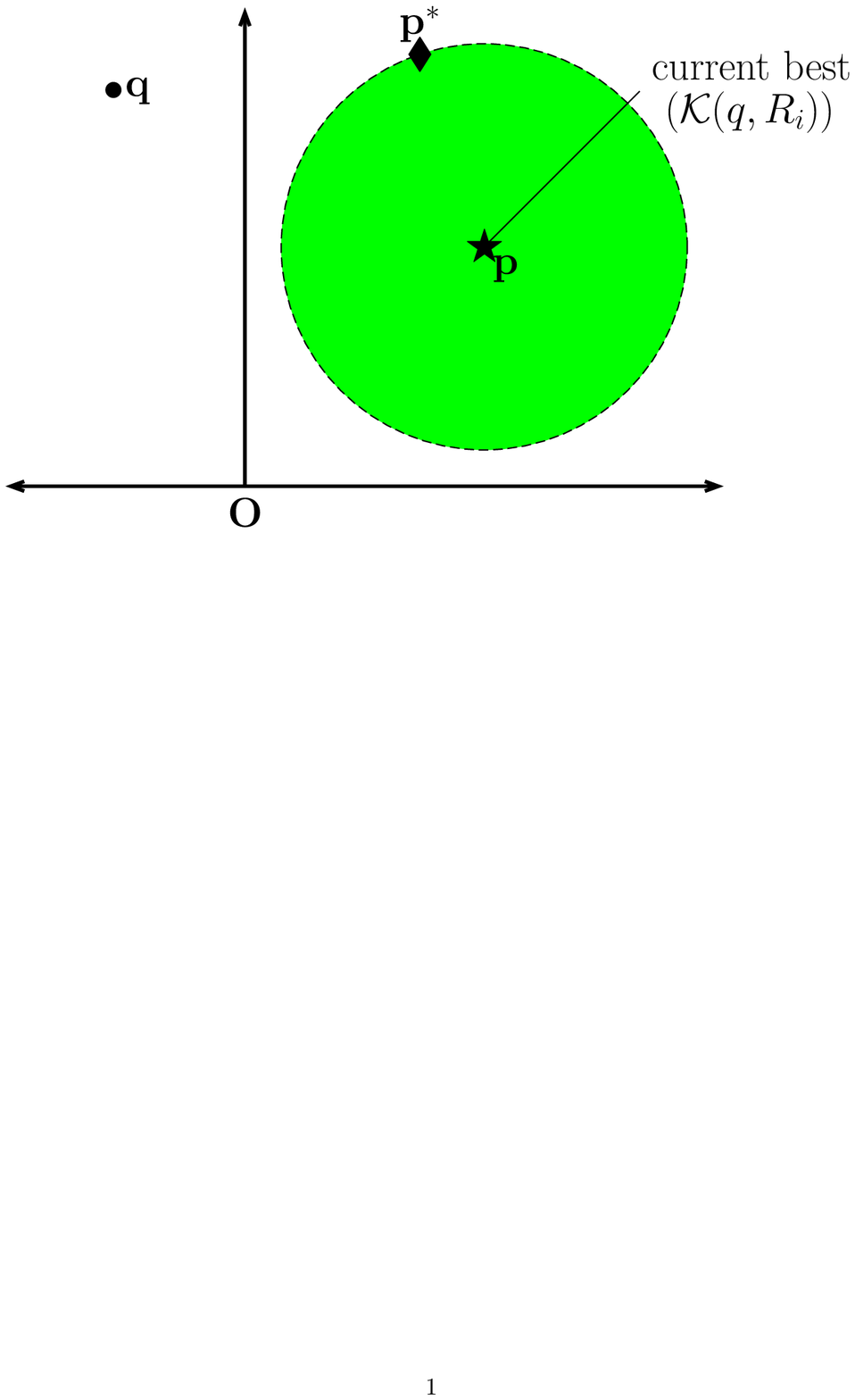}
}
\subfigure[{\small At level $i - 1$}]{ \label{fig:prune2}
\includegraphics[width=0.31\textwidth,clip=true,trim= 1.75in 6.03in 1.5in 1.7in]{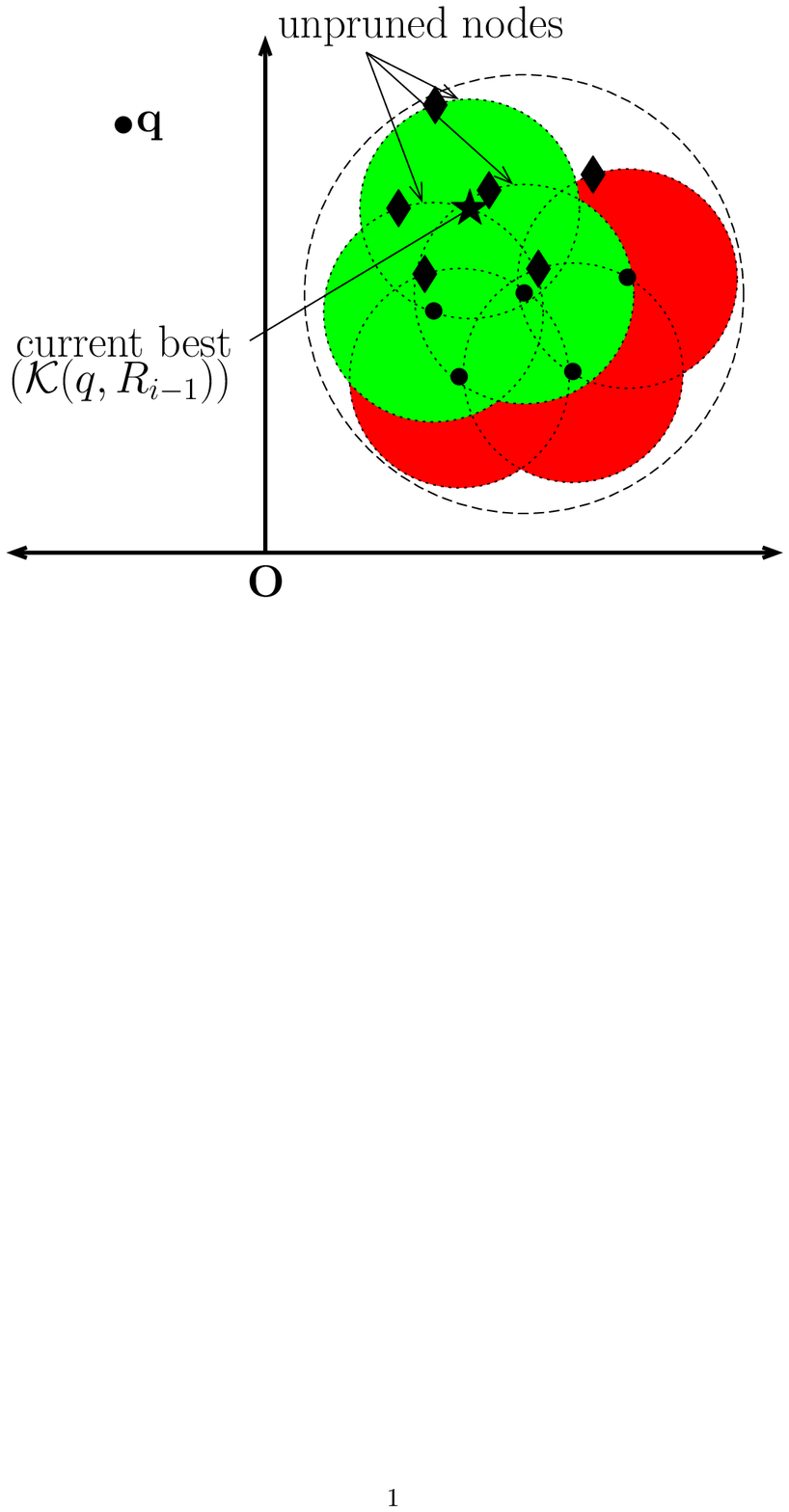}
}
\end{center}
\label{fig:pruning}
\mynegvspacex
\caption{Branch-and-bound tree-traversal -- the green nodes are retained and the red nodes are pruned.}
\mynegvspace
\mynegvspace
\mynegvspace
\end{figure*}
For normalized kernels $(\K(x,x) = 1 \forall\ x \in S)$, all the points are on the surface of a hyper-sphere in $\hilbert$. In this case, the above bound in Theorem \ref{thm:single-ball-bnd} is correct but possibly loose. In the following theorem, we present a tighter bound specifically for this condition:
\begin{theorem} \label{thm:single-ball-bnd-tight}
\mynegvspace
Consider a kernel $\K$ such that $\K(x,x) = 1 \forall\ x \in S$. Given a cover tree node rooted at an object $p$ at level $i$ in $\hilbert$ and a (query) object $q$, the maximum kernel function value between $q$ and any object in the set $S_p^i$ is bounded as:
\begin{equation} \label{eq:mko-ball-bnd-tight}
\mynegvspace
\K(q, S_p^i) \leq \begin{cases}
\begin{split} &\K(q, p)(1 - 2^{2i+1}) \\ & \quad + 2^{i+1}\sqrt{\left(1 - \K(q, p)^2 \right)\left(1 - 2^{2i}\right)}, \\ &\quad \mbox{ if } \K(q, p) \leq 1 - 2^{2i + 1} \end{split} \\
1.0, \mbox{ otherwise}
\end{cases}
\mynegvspace
\end{equation}
\mynegvspace
\mynegvspace
\end{theorem}
The proof is similar to the proof of Theorem \ref{thm:single-ball-bnd} and presented in Section B of the supplement.
\mynegvspace
\mynegvspace
\begin{algorithm}[H]
   \caption{\textbf{\underline{FastMKS}(Tree $T$, query $q$)}}
\begin{algorithmic}
\label{alg:single-fmko}
   \STATE \textbf{Initialize} $R_\infty = C_\infty$ \textit{// the root of the tree}
   \FOR{$i=\infty$ {\bfseries to} $-\infty$}
    \STATE $R = \{ Children(r)\colon r\in R_i\}$ \textit{// tree descend}
    \STATE $R_{i-1} = \{r\in R\colon \K(q,r) \geq K(q, R) - 2^i\sqrt{\K(q, q)} \}$
   \ENDFOR
   \STATE \textbf{return} $\arg\max_{r\in R_{-\infty}} \K(q,r)$ \textit{// the leaf nodes.}
\end{algorithmic}
\end{algorithm}
\mynegvspace
\mynegvspace

\noindent
\textbf{The branch-and-bound algorithm.} The bound on $\K(q, S_p^i)$ is used to decide whether a node is retained for further exploration or removed (``pruned'') from consideration. The branch-and-bound algorithm \textbf{FastMKS$(T, q)$} is presented in Alg. \ref{alg:single-fmko}. 
If the maximum possible kernel value from a subtree is less than the current best kernel value, that subtree is not explored any further. The best possible kernel value from a subtree and a step of the algorithm is depicted in Figure~\ref{fig:pruning}. For a retained node, all its children are explored.
The correctness of \textbf{FastMKS$(T, q)$} follows from:
\begin{theorem}\label{thm:bb-algo-correctness}
\mynegvspace
If $T$ is a cover tree on $S$, then \textbf{FastMKS$(T, q)$} (Alg. \ref{alg:single-fmko}) returns $\arg \max_{r\in S} \K(q, r)$.
\end{theorem}
\mynegvspace
We present a sketch proof. The complete proof is presented in Section C of the supplement.
\begin{proof}(sketch) By Theorem \ref{thm:single-ball-bnd}, $\K(q, R_{i -1}) \geq \K(q, S) - 2^i\sqrt{\K(q,q)} ~\forall\ i$. Therefore, $\lim\limits_{i\to -\infty} \K(q, R_{i-1}) = \K(q, S)$. Hence $\arg\max_{r\in R_{-\infty}} \K(q,r) = \arg\max_{r \in S} \K(q, r)$.
\end{proof}

\noindent
\textbf{Approximate max-kernel search.}  Similarity search problems can be approximated for further scalability. Even though we are focusing on exact max-kernel search in this paper, we wish to demonstrate that the tree based method can be very easily extended to perform the approximate max-kernel search. Approximation can be achieved in the following ways:
\mynegvspace
\begin{citemize}
\item[1.] Absolute value approximation (AVA): return $p\in S$ such that $\K(q, p) \geq \K(q,S) - \epsilon$.
\item[2.] Relative value approximation (RVA): return $p \in S$ such that $\K(q, p) \geq \K(q,S) - \epsilon |\K(q,S)|$.
\item[3.] Rank approximation (RA): return $p \in S$ such that $\left | \{r \in S \colon \K(q, r) > \K(q, p) \} \right| \leq \tau$.
\end{citemize}
\mynegvspace
Care has to be taken for relative value approximation since there is no guarantee that $\K(q,S) > 0$. The following stopping rules can be used at the end of an iteration in the for loop in \textbf{FastMKS$(T, q)$} (Alg. \ref{alg:single-fmko}) for the value approximations. The best candidate up until then is the approximate solution.
\mynegvspace
\begin{citemize}
\item[1.] AVA stopping: if $\epsilon \geq 2^i \sqrt{\K(q,q)}$, stop.
\item[2.] RVA stopping: assuming $\K(q, S) > 0$, if $\K(q, R_{i - 1}) \geq
\left(2^i / \epsilon\right) \sqrt{\K(q, q)}$, stop.\footnote{The stopping rule
can be easily modified for $\K(q,S) < 0$.}
\end{citemize}
\mynegvspace
\begin{theorem} \label{add-approx-correctness}
The AVA stopping rule returns a point $p \in S$ such that $\K(q, p) \geq \K(q, S) - \epsilon$.
\end{theorem}
\mynegvspace
\mynegvspace
\mynegvspace
\mynegvspace
\begin{proof}
At the end of iteration at level $i$, $\K(q, R_{i-1}) \geq \K(q, S) - 2^i \sqrt{\K(q,q)}$. If $2^i\sqrt{\K(q,q)} \leq \epsilon$, the AVA condition is satisfied.
\end{proof}
\mynegvspace
\mynegvspace
\mynegvspace
\begin{theorem} \label{mul-approx-correctness}
Assuming $\K(q,S) > 0$, the RVA stopping rule returns a point $p \in S$ such that $\K(q, p) \geq (1 - \epsilon) \K(q, S)$.
\end{theorem}
\mynegvspace
\mynegvspace
\mynegvspace
\mynegvspace
\begin{proof}
When $\K(q, R_{i-1}) \geq \left(2^i / \epsilon \right) \sqrt{\K(q,q)}$, then $2^i \sqrt{\K(q,q)} \leq \epsilon \K(q, R_{i-1}) \leq \epsilon \K(q,S)$. Using this, we have $\K(q, R_{i-1}) \geq \K(q, S) - 2^i \sqrt{\K(q,q)} \geq \K(q,S) - \epsilon \K(q,S)$.
\end{proof}
\mynegvspace
Rank-approximation can be achieved by performing stratified sampling on the cover tree \cite{ram2009rank}. The technique is more involved and presented in Section D of the supplement for the lack of space.

%% file: analysis.tex
The runtime analysis of \textbf{FastMKS} will make use of the following results \cite{beygelzimer2006cover}:
\mynegvspace
\begin{lemma}
\label{lma:wb}
The number of children of any node $p$ is bounded by $c^4$.
\end{lemma}
\mynegvspace
\mynegvspace
\begin{lemma}
\label{lma:db}
\mynegvspace
\mynegvspace
The maximum depth of any point $p$ in the explicit representation is $O(c^2 \log n)$.
\end{lemma}
\mynegvspace
\mynegvspace
The main result of this section is the search time complexity of \textbf{FastMKS$(T, q)$} in terms of the number of objects in $S$ and the properties of $(S, \K)$:
\begin{theorem} \label{thm:runtime}
\mynegvspace
Given a Mercer kernel $\K$, if the dataset $S$ of size $n$ has an expansion constant $c$ (with the metric $d_\K$) and a directional concentration constant $\gamma$, \textbf{FastMKS$(T, q)$} requires $\mathbf{O}(c^{12}\gamma^2 \log n)$ time.
\end{theorem}
\mynegvspace
\mynegvspace
\mynegvspace
\mynegvspace
\begin{proof}
The first part of the proof is similar to the runtime analysis of NNS with cover trees \cite{beygelzimer2006cover}. Let $R^*$ denote the last explicit $R_i$ considered by the algorithm. By Lemma \ref{lma:db}, the explicit depth of any point in $R^*$ is at most $k = O(c^2 \log n)$.
The maximum number of iterations required would be at most $k|R^*| \leq k \max_i |R_i|$.
The amount of work done in each iteration is at most $O(\max_i |R_i|)$, hence resulting in a total of $O(k \max_i |R_i|^2)$ work. Moreover, from Lemma \ref{lma:wb}, the total number of children encountered throughout the whole algorithm is at most $k \max_i |R_i| \cdot c^4$. Hence, the pruning/retaining rule does at most $O(k \max_i |R_i| \cdot c^4)$ work. Also, $R_{-\infty} \leq \max_i |R_i|$. Therefore, the algorithm requires at most $O(k \max_i |R_i|^2 + k \max_i |R_i| \cdot c^4)$ time.

Now we bound $\max_i |R_i|$. Let $u = \varphi(q) / \norm{\varphi(q)}$.
Then
$I_S(\varphi(q), [a, b]) = I_S\left(u, \left[a / \norm{\varphi(q)}, b/\norm{\varphi(q)}\right] \right)$.
For any level $i$, let $R = \{\mbox{Children}(r)\colon r \in R_i \}$ and let $\kappa = \K(q, R)$ and $\kappa^* = \K(q, S)$. Then
\begin{eqnarray*}
\mynegvspace
\mynegvspace
R_{i-1} & = & \{r \in R\colon \K(q, r) \geq \kappa - 2^i \norm{\varphi(q)} \} \\
    & = & I_S(\varphi(q), [\kappa - 2^i\norm{\varphi(q)}, \kappa^*]) \cap R \\
 & \subseteq & I_S(\varphi(q), [\kappa - 2^i\norm{\varphi(q)}, \kappa^*]) \cap C_{i-1} \\
    & \subseteq & I_S(\varphi(q), [\kappa^* - 2^{i+1}\norm{\varphi(q)}, \kappa^*]) \cap C_{i-1}
\mynegvspace
\mynegvspace
\end{eqnarray*}
since $\kappa^* \leq \kappa + 2^i \norm{\varphi(q)}$. 

\noindent
Then for any $r\in I_S(\varphi(q), [\kappa^* - 2^{i+1}\norm{\varphi(q)}, \kappa^*])$,
\begin{equation*} 
\mynegvspace
\begin{split}
& I_S(\varphi(q), [\kappa^* - 2^{i+1}\norm{\varphi(q)}, \kappa^*]) \\
   &\quad \subseteq  I_S(q, [\K(q,r) - 2^{i+1}\norm{\varphi(q)}, \K(q,r) + 2^{i+1}\norm{\varphi(q)}]) \\
 &\quad = I_S(u, [\ip{u}{\varphi(r)} - 2^{i+1}, \ip{u}{\varphi(r)} + 2^{i+1}]).
\end{split}
\end{equation*}
By the definition of the directional concentration constant, there exists $r_j$s such that:
\begin{equation*} 
\mynegvspace
\begin{split}
& I_S(u, [\ip{u}{\varphi(r)} - 2^{i+1}, \ip{u}{\varphi(r)} + 2^{i+1}]) \\
&\quad \subseteq \cup_{j=1}^\gamma B_S(r_j, 2^{i+1}).
\end{split} 
\mynegvspace
\end{equation*}
Bounding the size of $I_S(q, [K - 2^i\norm{\varphi(q)}, \K(q,S)]) \cap C_{i-1}$ amounts to bounding the number of disjoint balls of radius $2^{i-2}$ that can be packed into each of the $\gamma$ balls $B_S(r_j, 2^{i+1} + 2^{i-2})$. For each of the $r_j$, we have:
\begin{equation*} \mynegvspace 
\begin{split}
& |B_S(r_j, 2^{i+1} + 2^{i-2})| \leq |B_S(r', 2^{i+2} + 2^{i-1})| \\
&\quad \leq |B_S(r', 2^{i+3})| \leq c^5|B_S(r, 2^{i-2})|.
\end{split}
\mynegvspace
\end{equation*}
Hence, $\forall\ i$, $|R_i| \leq \gamma c^5$. Hence $\max_i |R_i| \leq \gamma c^5$, thus giving us the statement of the theorem.
\end{proof}
Comparing to the query time $O(c^{12} \log n)$ for NNS
\cite{beygelzimer2006cover}, it is clear that \textbf{FastMKS} has similar $\log
n$ scaling, but also has an extra price of $\gamma^2$ for solving the more general problem of max-kernel search.

%% file: eval.tex
We evaluate the \textbf{FastMKS} algorithm with different kernels and datasets.
For every experiment, we query the top $\{1, 2, 5, 10\}$ max-kernel candidates
and report the speedup over linear search (in terms of the number of kernel
evaluations performed). The cover tree and the algorithm is developed in
MLPACK~\cite{curtin2011mlpack} with the implementation details in Section E of
the supplement.\footnote{See http://www.mlpack.org/ for more information on
MLPACK.}
\begin{figure*}[t]
\begin{center}
\centerline{
\subfigure[Cosine and hyperbolic tangent kernels]{
\includegraphics[height=0.47\textwidth, angle = -90]{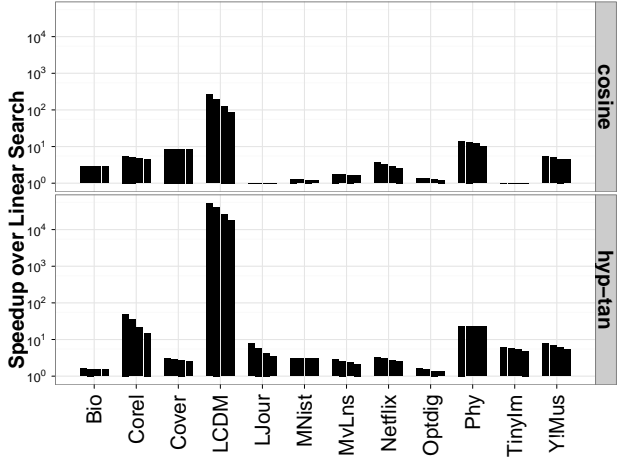}
}
\subfigure[Linear and polynomial (degree 10) kernels]{
\includegraphics[height=0.47\textwidth, angle = -90]{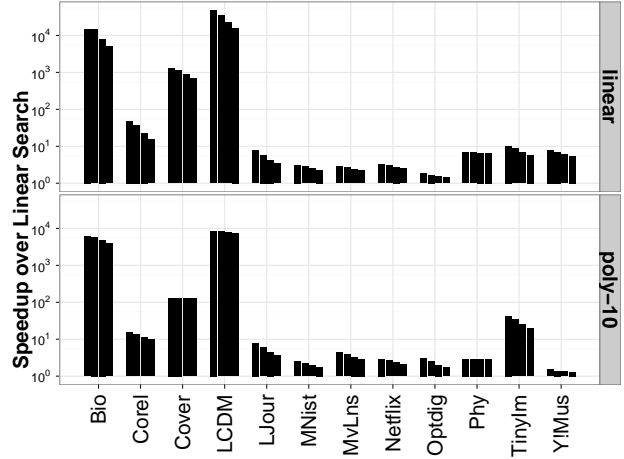}
}
}
\end{center}
\mynegvspacex
\mynegvspace
\caption{Speedups over linear scan for \textbf{FastMKS} with $k = 1,2,5,10$.}
\label{fig:speedups-all}
\mynegvspace
\mynegvspace
\mynegvspace
\end{figure*}

\vspace{0.05in}
\noindent
\textbf{Datasets.} We use two different classes of datasets -- (1) Datasets with fixed-length objects: These include the MNIST dataset \cite{lecun2000mnist}, the Isomap ``Images'' dataset, several datasets from the UCI machine learning repository \cite{ucimlrepository}, three collaborative filtering datasets (MovieLens, Netflix \cite{bennett2007netflix}, Yahoo! Music
\cite{dror2011yahoo}), the LCDM astronomy dataset \cite{lupton2001sdss}, the LiveJournal blog moods text dataset \cite{Kim2011pre} and a subset of the 80 million tiny images dataset \cite{torralba200880} (the sizes of the datasets are presented in Table~\ref{tab:datasets}). (2) Datasets without fixed length representation: We consider protein sequences from the GenBank\footnote{ftp://ftp.ncbi.nih.gov/refseq/release/complete/}.
\begin{table}[h]
\begin{center}
{\small
\begin{tabular}{|l|c|c|c|}
\hline
Datasets & $|Q|$ & $|S|$ & dims \\ \hline
\textbf{Y! Music} & 100000 & 600000 & 50 \\ \hline
\textbf{MovieLens} & 6040 & 3706 & 50 \\ \hline
\textbf{Opt-digits} & 450 & 1347 & 64 \\ \hline
\textbf{Physics} & 37500 & 112500 & 78 \\ \hline
\textbf{Bio} & 75000 & 210409 & 74 \\ \hline
\textbf{Covertype} & 150000 & 431012 & 55 \\ \hline
\textbf{LiveJournal} & 10000 & 10000 & 25327 \\ \hline
\textbf{MNIST} & 10000 & 60000 & 784 \\ \hline
\textbf{Netflix} &  480189 & 17770 & 50 \\ \hline
\textbf{Corel} & 10000 & 27749 & 32 \\ \hline
\textbf{LCDM} & 6000000 & 10777216 & 3 \\ \hline
\textbf{TinyImages} & 5000 & 1000000 & 384 \\ \hline
\end{tabular}
}
\label{tab:datasets}
\end{center}
\mynegvspacex
\caption{Details of the vector datasets.}
\mynegvspace
\mynegvspace
\mynegvspace
\end{table}

\vspace{0.05in}
\noindent
\textbf{Kernels.} We consider all of the following kernels for the vector
datasets: \textit{cosine}, \textit{hyperbolic tangent}\footnote{While the
hyperbolic tangent kernel is not Mercer in general, our proposed algorithm works
correctly with non-Mercer kernels if the kernel is positive definite when
restricted to the dataset.}, \textit{linear}, and \textit{polynomial} (with
degree 10). While \textbf{FastMKS} is applicable to any kernel, MKS with the
Gaussian kernel reduces to NNS in the input space; hence, we omit this kernel from our experiments. The $p$-spectrum kernel \cite{leslie2002spectrum} is used for the protein
sequence data.

\noindent
\textbf{Results.} The results for the vector datasets are summarized in Figure
\ref{fig:speedups-all}. While the speedups range from anywhere between $1$ to
over $10^4$, a speedup close to an order of magnitude is seen in most large
datasets (except MNIST). It is also quite clear that different kernels give very
different speedups for the same dataset. This can be attributed to the fact that
the expansion constant and the directional concentration constant are properties
of the dataset-kernel pair. The results for the protein sequence data (Figure
\ref{fig:protein-seq-scaling}) indicate that the speedups increase with
increasing reference set size, recording a speedup of over two orders of
magnitude for a set of around $100000$ sequences, exhibiting the logarithmic
scaling of \textbf{FastMKS}.
\begin{figure}[h]
\mynegvspace
\begin{center}
\includegraphics[height=0.85\columnwidth, angle = -90]{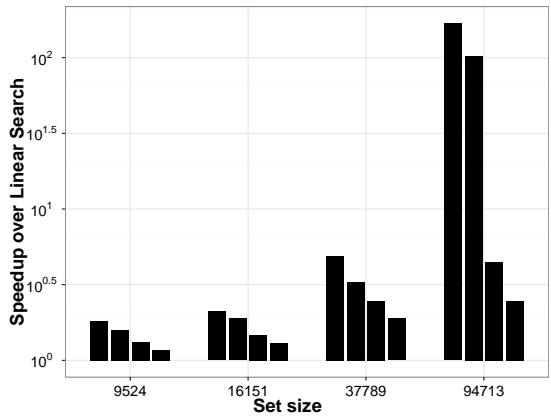}
\end{center}
\mynegvspacex
\caption{Speedups over linear scan for \textbf{FastMKS} with $k = 1,2,5,10$ with sets of protein sequences of increasing size to demonstrate the scaling of \textbf{FastMKS}.}
\label{fig:protein-seq-scaling}
\mynegvspace
\mynegvspace
\mynegvspace
\end{figure}

%% file: par.tex
Despite the theoretical and empirical efficiency of \textbf{FastMKS}, the algorithm still requires the dataset (and the tree) to fit completely in the memory of a single machine. Large datasets generally need to be distributed across multiple nodes to load the whole data in memory. While a tree can be constructed on distributed data, this tree is inefficient for MKS because of the high internode communication in the branch-and-bound algorithm. In this section, we discuss a method to extend \textbf{FastMKS} to the distributed data setting and compare it to the obvious baseline. Since communication and data exchange are major bottlenecks in parallel and distributed systems, our proposed method and the baseline avoid these bottlenecks as well as fit into the map-reduce framework.

\vspace{0.05in}
\noindent
\textbf{Linear scan baseline.} Indexing schemes are generally not conducive to distributed data and distributed linear scan is a possible alternative. Let the dataset of size $n$ be evenly distributed over $m$ nodes with each node containing $(n/m)$ objects and let there be a master node that communicates with
each of these $m$ nodes. For single query, the master sends the query to each of the $m$ nodes in the map phase and the mappers on every node perform a linear scan in parallel. The best match from each of these nodes is returned to the master node which can then return the best among those $m$ returned matches in the reduce phase. Then the total time required to obtain the best match is $\left( c_1 m + c_2 \frac{n}{m} \right)$. It is important to note that each of the $m$ nodes containing the actual data performs $(n/m)$ kernel evaluations while the master node just sorts $m$ kernel values; therefore, $c_1 \ll c_2$. The minimum possible runtime achievable by linear scan in this distributed system is
\mynegvspace
\begin{equation*}
\mynegvspace
2\sqrt{c_1 c_2 n} \sim O(\sqrt{n}) \mbox{ with } m = \sqrt{c_2 n / c_1}.
\mynegvspace
\end{equation*}

\vspace{0.05in}
\noindent
\textbf{Multiple tree indices.} Let us consider this following scheme for indexing: for each of the $m$ nodes containing $(n/m)$ objects, build and save a cover tree on the data in each of the nodes. Similar to the linear scan case, for a single query, the master sends the query to each of the $m$ nodes in the map phase and the mappers on every node perform the branch-and-bound search on the cover trees in parallel. The best matches from each of $m$ trees are returned to the master node which then reports the best among them in the reduce phase. Since we have shown that a single query on a cover tree requires logarithmic time, let the time taken for this parallel search on multiple trees be $\left(c_1 m + c'_2 \log \frac{n}{m}\right),$ where $c'_2$ is the scaling constant independent of $n$ for the cover tree on this dataset. In this scenario, the minimum possible achievable runtime is
\mynegvspace
\begin{equation*}
c'_2 \log ( e\ n\ c_1 / c'_2) \sim O(\log n) \mbox{ with } m = c'_2 / c_1.
\mynegvspace
\end{equation*}

\noindent
\textbf{Remark.} The optimal values of $m$ can only be used if there are enough number of nodes available and if $(n / m)$ is small enough for the data in the node to fit in memory. However, searching over multiple tree indices is a simple yet efficient extension of our proposed method for distributed data. In the best case, our method can achieve $O(\log n)$ scaling while the distributed linear scan can only achieve $O(\sqrt{n})$ scaling. More non-trivial extensions using recent techniques \cite{cayton2012accelerating, lee2012distributed} can achieve further efficiency in the distributed setting.

%% file: conclusion.tex
%
%
Our work appears to be the most comprehensive study of the general MKS problem to date, including the first rigorous characterization of its hardness and the first general-kernel method with provably logarithmic query time, sub-quadratic preprocessing time, and speed in practice. A tighter theoretical analysis without the bounded directional concentration constant assumption and an empirical investigation on more abstract objects like graphs and time series will be part of our future work. We wish to extend our search algorithm to more classes of similarity functions beyond Mercer kernels.
%

%% file: app.tex
%
%
\section{Cover tree construction \\ in Hilbert space}
The details of the cover tree construction are provided in Algorithm \ref{alg:tree-construction}. For a cover tree built on dataset $S$, each node is only associated with a single point $p \in S$.  Therefore, because the only distance computations involve points in $S$, the explicit representation of the objects in $\hilbert$ is not required (by the kernel trick). 

Given a Mercer kernel $\K$ for a class of objects, the distances between points in $\hilbert$ required for the tree construction in $\hilbert$ can be evaluated using the distance metric $d_\K$ induced from the kernel.  Three kernel evaluations are required to compute the distance; however, if the self-kernel values $\K(x, x) \forall\ x \in S$ are precomputed and saved, $d_\K(x, y)$ can be evaluated with a single evaluation of $\K(x, y)$.

\begin{algorithm}[H]
\caption{\textbf{\underline{Construct}}$(p,\left\langle\mbox{\textsc{Near}, \textsc{Far}}\right\rangle, i)$ \cite{beygelzimer2006coverLonger}}
\begin{algorithmic}
{\small
\label{alg:tree-construction}
  \IF{$\mbox{\textsc{Near}}= \emptyset$}
    \STATE return $\{ p, \emptyset \}$.
  \ELSE
    \STATE $\{\mbox{\textsc{Self, Near}}\}$ = \\
\ \ \ \ \ \textbf{Construct}$(p, \mbox{\textbf{Split}}(d(p, \cdot), 2^{i-1}, \{\mbox{\textsc{Near}}\}), i - 1)$
    \STATE add \textsc{Self} to Children$(p)$
    \WHILE{$\mbox{\textsc{Near}}\not= \emptyset$}
      \STATE pick $q$ in \textsc{Near}
      \STATE $\{\mbox{\textsc{Child, Unused}}\}$ = \\
       \ \ \ \ \textbf{Construct}$(q, \mbox{\textbf{Split}}(d(q, \cdot), 2^{i-1}, \{\mbox{\textsc{Near, Far}}\}),$\\ \ \ \ \ \ \ \ \ \ \ \ \ $ i - 1)$
      \STATE add \textsc{Child} to Children$(p)$
      \STATE $\{\mbox{\textsc{New-near, New-far}}\}$ = \\
      \ \ \ \ \ \ \textbf{Split}$(d(p, \cdot), 2^i, \{\mbox{\textsc{Unused}}\})$
      \STATE $\textsc{Far}\leftarrow \textsc{Far}\cup\textsc{New-far}$
      \STATE $\textsc{Near}\leftarrow \textsc{Near}\cup\textsc{New-near}$
    \ENDWHILE
    \STATE return $\{ p, \textsc{Far} \}$.
  \ENDIF
}
\end{algorithmic}
\textbf{\underline{Split}}$(d(p, \cdot), r, \{S_1, S_2, \ldots\})$
\begin{algorithmic}
\STATE \textsc{Near} = $\bigcup_{i}\{q \in S_i \colon d(p, q) \leq r \}$
\STATE \textsc{Far} = $\bigcup_{i}\{q \in S_i \colon 2r > d(p, q) > r \}$
\STATE $\forall\ i, S_i \leftarrow S_i \setminus (\mbox{\textsc{Near}} \cup \mbox{\textsc{Far}})$.
\STATE return $\{ \mbox{\textsc{Near, Far}} \}.$
\end{algorithmic}
\end{algorithm}
The algorithm calls the recursive function \textbf{Construct}$(p, \left\langle\mbox{\textsc{Near}, \textsc{Far}}\right\rangle, i)$ where $p$ is a point, $\left\langle\mbox{\textsc{Near}, \textsc{Far}}\right\rangle$ are the point sets and $i$ is the current level of the tree. This recursive function calls a subroutine \textbf{Split}$(d(p, \cdot), r, \{S_1, S_2, \ldots\})$ with $d(p, \cdot)$ representing the set of distances of the points in the point sets $\{S_1, S_2, \ldots\}$ to the point $p$ and $r$ is the splitting distance to form the \textsc{Near} and \textsc{Far} sets.
\section{Proof of Theorem 4.2}
We restate the theorem for completeness:
\begin{theorem}
Consider a kernel $\K$ such that $\K(x,x) = 1 \forall\ x \in S$. Given a cover tree node rooted at an object $p$ at level $i$ in $\hilbert$ and a (query) object $q$, the maximum kernel function value between $q$ and any object in the set $S_p^i$ is bounded as:
\begin{equation} \label{eq:mko-ball-bnd-tight2}
\K(q, S_p^i) \leq \begin{cases}
\begin{split} &\K(q, p)(1 - 2^{2i+1}) \\ & \quad + 2^{i+1}\sqrt{\left(1 - \K(q, p)^2 \right)\left(1 - 2^{2i}\right)}, \\ &\quad \mbox{ if } \K(q, p) \leq 1 - 2^{2i + 1} \end{split} \\
1.0, \mbox{ otherwise}
\end{cases}
\end{equation}
\end{theorem}
\begin{proof}
Since all the points are sitting on the surface of a hypersphere in $\hilbert$, $\K(q, p)$ denotes the cosine of the angle made by $\varphi(q)$ and $\varphi(p)$ at the origin.  If we first consider the case where $q$ lies within the ball bounding cover tree node $p$ at level $i$ (that is, if $d_K(q, p) < 2^{i + 1}$), it is clear that the maximum possible kernel evaluation should be 1, because there could exist a point in $S_p^i$ whose angle to $q$ is 0.  We can easily modify this condition to an easier condition on $\K(q, p)$, which is $\K(q, p) < 1 - 2^{2i + 1}$.

Now, for the other case, let $\cos \theta_{qp} = \K(q, p)$ and $p^* = \arg \max_{r \in S_p^i} \K(q, r).$ Let $\theta_{pp^*}$ be the angle between $\varphi(p)$ and $\varphi(p^*)$ and $\theta_{qp^*}$ be the angle between $\varphi(q)$ and $\varphi(p^*)$ at the origin. Then 
\begin{equation} \label{eq:mko-cos-bnd1}
\K(q, p^*)  =  \cos \theta_{qp^*} \leq \cos (\{\theta_{qp} - \theta_{pp^*}\}_+).
\end{equation}
Now we know that $d_\K(p, p^*) \leq 2^{i+1}$ and that $d_\K(p, p^*) = \sqrt{2 - 2 \cos \theta_{pp^*}}.$ Hence $\cos \theta_{pp^*} \geq 1 - 2^{2i+1},$ and $\theta_{pp^*} \leq |\cos^{-1} (1 - 2^{2i+1})|$ and take $\theta = |\cos^{-1} (1 - 2^{2i + 1})|$. Combining this with equation \eqref{eq:mko-cos-bnd1}, we get:
\begin{equation*} \label{eq:mko-cos-bnd2}
\K(q, S_p^i) \leq \cos (\{\theta_{qp} - \theta_{pp^*}\}_+) \leq \cos (\{\theta_{qp} - \theta\}_+).
\end{equation*}
Substituting the value of $\theta$ above and simplifying gives us the statement of the theorem.
\end{proof}
\section{Proof of Theorem 4.3}
We will restate the theorem here: 
\begin{theorem}
If $T$ is a cover tree on $S$, then \textbf{FastMKS$(T, q)$} returns $\arg \max_{r\in S} \K(q, r)$.
\end{theorem}
\mynegvspace
\begin{proof}
Let $p^* = \arg \max_{r \in S} \K(q, r)$ and $\kappa^* = \K(q, p^*)$. At the beginning of the iteration at level $\infty$, $p^*$ is in consideration since $p^*\in S$ and for $p \in C_\infty$, $S_p^\infty = S$. At the end of the iteration at level $i$, all points $p \in S$ such $\K(q, p) \geq \K(q, R_{i-1}) - 2^i \sqrt{\K(q,q)}$ are still in consideration.  Since $\kappa^* \geq \K(q, R_{i-1})$, $p^*$ is still in consideration. Either $p^* \in R_{i-1}$ or $p^*$ is the grandchild of some point $p \in R_{i - 1}$. Hence, by theorem 4.1, $\K(q, R_{i -1}) \geq \kappa^* - 2^i\sqrt{\K(q,q)} ~\forall\ i$. Now $\lim_{i\to -\infty} \K(q, R_{i-1}) \geq \lim_{i\to -\infty} \left(\kappa^* - 2^i \sqrt{\K(q,q)} \right) = \kappa^*$ (assuming $\K(q,q) < \infty$). Hence $\K(q, R_{-\infty}) = \kappa^*$ and \textbf{FastMKS$(T,q)$} returns $\arg\max_{r\in R_{-\infty}} \K(q,r) = p^*$.
\end{proof}
\section{Rank-approximate \\max-kernel search}
The rank-approximation of the max-kernel search is defined as follows: for a given set $S$ of $n$ objects (the reference set), a (Mercer) kernel function $\K(\cdot, \cdot)$, and a query $q$, find the object $p \in S$ such that $$\left | \{r \in S \colon \K(q, r) > \K(q, p) \} \right| \leq \tau.$$
As mentioned in Ram et.al, 2009 \cite{ram2009rank}, the idea is to draw enough samples $S'$ from the tree such that $$\Pr\left(|\{r \in S \colon \K(q,r) > \K(q, S') \}| < \tau\right) \geq 1 - \delta.$$
Simplifying the formulation presented in \cite{ram2009rank}, the probability of always missing the top $\tau$ values for a given query $q$ after $k$ samples with replacement is given by $\left(1 - \frac{\tau}{n} \right)^k.$ If we want a $(1 - \delta)$ success rate of sampling, then we want $k$ to be such that 
$$\left(1 - \frac{\tau}{n} \right)^k < \delta, \mbox{\ \ \ and \ \ \ }  \left(1 - \frac{\tau}{n} \right)^{k-1} > \delta,$$
giving $k = \left\lceil \frac{\log \delta}{\log \left(1 - \frac{\tau}{n}\right)} \right\rceil.$ Given that $k$ samples (as defined above) is to be made from the tree, the stratified sampling on a tree is presented in Algorithm \ref{alg:single-ra-fmko}. We assume that at each node of the tree, we have access to the number of points in the subtree $S_r^i$ rooted at node $r$ at level $i$ for every node in the tree. This algorithm returns a $\tau$-rank approximate solution to the max-kernel operation with probability $(1 - \delta)$.
\begin{algorithm}[H]
   \caption{\textbf{\underline{RAFastMKS}(Tree $T$, query $q$, rank error $\tau$, failure probability $\delta$)}}
\begin{algorithmic}{\small
\label{alg:single-ra-fmko}
   \STATE \textbf{Initialize} $R_\infty = C_\infty$ \hfill \textit{// the root of the tree}
   \STATE \textbf{Set} $k = \left\lceil \frac{\log \delta}{\log \left(1 - \frac{\tau}{n}\right)} \right\rceil$ \hfill \textit{// the number of samples}
   \FOR{$i=\infty$ {\bfseries to} $-\infty$}
    \STATE $R = \{ Children(r)\colon r\in R_i\}$ \hfill \textit{// tree descend}
    \STATE $R' = \{r\in R\colon \K(q,r) \geq K(q, R) - 2^i\sqrt{\K(q, q)} \}$
    \STATE $R_{i-1} = \{r \in R' \colon |S_r^i| > \frac{n}{k} \}$ \textit{// approximate by sampling}
    \STATE $R_{-\infty} = R_{-\infty} \cup \{R' \setminus R_{i-1} \}.$
   \ENDFOR
   \STATE \textbf{return} $\arg\max\limits_{r\in R_{-\infty}} \K(q,r)$ \hfill \textit{// the leaf nodes.}
}\end{algorithmic}
\end{algorithm} 

\section{Implementation details} 

We use the following performance-improving optimizations on the cover tree: 
\begin{citemize}
\item[(1)] Instead of the upperbound of $2^{i+1}$ for the distance of a cover tree node $p$ at level $i$ to its furthest descendant, the actual distance to the furthest descendant is cached at tree construction time and used at search time, 
\item[(2)] The distance to the parent is cached to avoid evaluating $\K$ for the child nodes where an improvement is impossible,
\item[(3)] Instead of the base $2$, an experimentally verified base of $1.3$ is used for best results \cite{beygelzimer2006coverLonger}, 
\item[(4)] $\K(x,x) \forall x\in S$ is precomputed and stored for future use to reduce the number of kernel evaluations required for a single distance computation $d_\K(x,y)$ to one (just $\K(x,y)$ since $\K(x,x)$ and $\K(y,y)$ are precomputed), 
\item[(5)] We use the tighter bound in Theorem 4.2 for normalized kernels.
\end{citemize}